\documentclass[a4paper,UKenglish]{lipics-v2016}
\usepackage{amsmath,amsthm}

\usepackage{amssymb,calc}
\usepackage{mdframed}
\usepackage{microtype}
\usepackage{hyperref}
\usepackage{graphicx}
\usepackage{todonotes}
\title{A $2\ell k$ Kernel for $\ell$-Component Order Connectivity}
\titlerunning{A $2\ell k$ Kernel for $\ell$-Component Order Connectivity} 

\author[1]{Mithilesh Kumar}
\author[2]{Daniel Lokshtanov}
\affil[1]{Department of Informatics, University of Bergen\\
  Norway\\
  \texttt{Mithilesh.Kumar@ii.uib.no}}
\affil[2]{Department of Informatics, University of Bergen\\
  Norway\\
  \texttt{daniello@ii.uib.no}}
\authorrunning{M. Kumar and D. Lokshtanov} 

\Copyright{Mithilesh Kumar and Daniel Lokshtanov}


\EventEditors{John Q. Open and Joan R. Acces}
\EventNoEds{2}
\EventLongTitle{42nd Conference on Very Important Topics (CVIT 2016)}
\EventShortTitle{CVIT 2016}
\EventAcronym{CVIT}
\EventYear{2016}
\EventDate{December 24--27, 2016}
\EventLocation{Little Whinging, United Kingdom}
\EventLogo{}
\SeriesVolume{42}
\ArticleNo{23}

\newcommand{\Oh}{\mathcal{O}}
\newcommand{\pol}{n^{\mathcal{O}(1)}}
\newcommand{\cclass}[1]{\textsf{\textup{#1}}}
\newcommand{\NP}{\cclass{NP}}

\begin{document}
\maketitle
\begin{abstract}

In the $\ell$-\textsc{Component Order Connectivity} problem ($\ell \in \mathbb{N}$), we are given a graph $G$ on $n$ vertices, $m$ edges and a non-negative integer $k$ and asks whether there exists a set of vertices $S\subseteq V(G)$ such that $|S|\leq k$ and the size of the largest connected component in $G-S$ is at most $\ell$. 
In this paper, we give a linear programming based kernel for $\ell$-\textsc{Component Order Connectivity} with at most $2\ell k$ vertices that takes $n^{\mathcal{O}(\ell)}$ time for every constant $\ell$. Thereafter, we provide a separation oracle for the LP of \(\ell\)-COC implying that the kernel only takes $(3e)^{\ell}\cdot\pol$ time. On the way to obtaining our kernel, we prove a generalization of the $q$-Expansion Lemma to weighted graphs. This generalization may be of independent interest.

\end{abstract}
\section{Introduction}
In the classic {\sc Vertex Cover} problem, the input is a graph $G$ and integer $k$, and the task is to determine whether there exists a vertex set $S$ of size at most $k$ such that every edge in $G$ has at least one endpoint in $S$. Such a set is called a {\em vertex cover} of the input graph $G$. An equivalent definition of a vertex cover is that every connected component of $G - S$ has at most $1$ vertex. This view of the {\sc Vertex Cover} problem gives rise to a natural generalization: can we delete at most $k$ vertices from $G$ such that every connected component in the resulting graph has at most $\ell$ vertices? Here we study this generalization. Formally, for every integer $\ell \geq 1$, we consider the following problem, called \textsc{$\ell$-Component Order Connectivity} ($\ell$-\textsf{COC}).

\smallskip
\noindent
\fbox{\parbox{\textwidth-\fboxsep}{
\textsc{$\ell$-Component Order Connectivity} ($\ell$-\textsf{COC})\\
\textbf{Input:} A graph $G$ on $n$ vertices and $m$ edges, and a positive integer $k$.\\
\textbf{Task:} determine whether there exists a set $S\subseteq V(G)$ such that $|S|\leq k$ and the maximum size of a component in $G-S$ is at most $\ell$.
}}
\smallskip

The set $S$ is called an $\ell$-COC {\em solution}. For $\ell = 1$, $\ell$-\textsf{COC} is just the {\sc Vertex Cover} problem. Aside from being a natural generalization of {\sc Vertex Cover}, the family $\{\ell$-\textsf{COC} $: \ell \geq 1\}$ of problems can be thought of as a vulnerability measure of the graph $G$ - how many vertices of $G$ have to fail for the graph to break into small connected components? For a study of $\ell$-\textsf{COC} from this perspective see the survey of Gross et al.~\cite{gross}.

From the work of Lewis and Yannakakis~\cite{LewisY80} it immediately follows that $\ell$-\textsf{COC} is NP-complete for every $\ell \geq 1$. This motivates the study of $\ell$-\textsf{COC} within paradigms for coping with NP-hardness, such as approximation algorithms~\cite{approx_book}, exact exponential time algorithms~\cite{exact_book}, parameterized algorithms~\cite{pc_book,DowneyF13book} and kernelization~\cite{kernel_survey_kratsch, kernel_survey_loks}. The $\ell$-\textsf{COC} problems have (for some values of $\ell$) been studied within all four paradigms, see the related work section. 

In this work we focus on $\ell$-\textsf{COC} from the perspective of parameterized complexity and kernelization. Our main result is an algorithm that given an instance $(G, k)$ of $\ell$-\textsf{COC}, runs in polynomial time, and outputs an equivalent instance $(G', k')$ such that $k' \leq k$ and $|V(G')| \leq 2\ell k$. This is called a {\em kernel} for $\ell$-\textsf{COC} with $2\ell k$ vertices. Our kernel significantly improves over the previously best known kernel with $O(\ell k(k+\ell))$ vertices by Drange et al.~\cite{DrangeDH14}. Indeed, for $\ell = 1$ our kernel matches the size of the smallest known kernel for {\sc Vertex Cover}~\cite{ChenKJ01} that is based on the classic theorem of Nemhauser and Trotter~\cite{NemhauserT74}. 

\medskip
\noindent
{\bf Related Work.} $1$-\textsf{COC}, better known as {\sc Vertex Cover}, is extremely well studied from the perspective of approximation algorithms~\cite{approx_book,dinur2005hardness}, exact exponential time algorithms~\cite{FominGK09,Robson86,XiaoN13}, parameterized algorithms~\cite{pc_book,ChenKX10} and kernelization~\cite{ChenKJ01,NemhauserT74}. The kernel with $2k$ vertices for {\sc Vertex Cover} is considered one of the fundamental results in the field of kernelization. The $2$-\textsf{COC} problem is also well studied, and has been considered under several different names. The problem, or rather the dual problem of finding a largest possible set $S$ that induces a subgraph in which every connected component has order at most $2$, was first defined by Yannakakis~\cite{Yannakakis81a} under the name Dissociation Set. The problem has attracted attention in exact exponential time algorithms~\cite{KardosKS11,XiaoK15}, the fastest currently known algorithm~\cite{XiaoK15} has running time $O(1.3659^n)$. $2$-\textsf{COC} has also been studied from the perspective of parameterized algorithms~\cite{ChangCHRS16,Tu15} (under the name {\sc Vertex Cover} $P_3$) as well as approximation algorithms~\cite{Tu}. The fastest known parameterized algorithm, due to Chang et al.~\cite{ChangCHRS16} has running time $1.7485^kn^{O(1)}$, while the best approximation algorithm, due to Tu and Zhou~\cite{Tu} has factor $2$. 

For the general case of $\ell$-\textsf{COC}, $\ell \geq 1$, Drange et al.~\cite{DrangeDH14} gave a simple parameterized algorithm with running time $(\ell + 1)^kn^{O(1)}$, and a kernel with $O(k\ell(\ell+k))$ vertices. The parameterized algorithm of Drange et al.~\cite{DrangeDH14} can be improved to $(\ell + 0.0755)^kn^{O(1)}$ by reducing to the $(\ell + 1)$-{\sc Hitting Set} problem, and applying the iterative compression based algorithm for $(\ell + 1)$-{\sc Hitting Set} due to Fomin et al.~\cite{FominGKLS10}. The reduction to $(\ell + 1)$-{\sc Hitting Set}, coupled with the simple factor $(\ell + 1)$-approximation algorithm for  $(\ell + 1)$-{\sc Hitting Set}~\cite{approx_book} immediately also yields an $(\ell + 1)$-approximation algorithm for $\ell$-\textsf{COC}. There has also been some work on $\ell$-\textsf{COC} when the input graph is restricted to belong to a graph class, for a discussion of this work see~\cite{DrangeDH14}.

Comparing the existing results with our work, we see that our kernel improves over the kernel of Drange et al.~\cite{DrangeDH14} from at most $O(k\ell(\ell+k))$ vertices to at most $2k\ell$ vertices. Our kernel is also the first kernel with a linear number of vertices for every fixed $\ell \geq 2$.

\medskip
\noindent
{\bf Our Methods.} Our kernel for $\ell$-\textsf{COC} hinges on the concept of a {\em reducible pair} of vertex sets. Essentially ({\em this is not the formal definition used in the paper!}), a reducible pair is a pair $(X,Y)$ of disjoint subsets of $V(G)$ such that $N(Y) \subseteq X$, every connected component of $G[Y]$ has size at most $\ell$, and every solution $S$ to $G$ has to contain at least $|X|$ vertices from $G[X \cup Y]$. If a reducible pair is identified, it is easy to see that one might just as well pick all of $X$ into the solution $S$, since any solution has to pay $|X|$ inside $G[X \cup Y]$, and after $X$ is deleted, $Y$ breaks down into components of size at most $\ell$ and is completely eliminated from the graph.

At this point there are several questions. (a) How does one argue that a reducible pair is in fact reducible? That is, how can we prove that any solution has to contain at least $|X|$ vertices from $X \cup Y$? (b) How big does $G$ have to be compared to $k$ before we can assert the existence of a reducible pair? Finally, (c) even if we can assert that $G$ contains a reducible pair, how can we find one in polynomial time?

To answer (a) we restrict ourselves to reducible pairs with the additional property that each connected component $C$ of $G[Y]$ can be assigned to a vertex $x \in N(C)$, such that for every $x \in X$ the total size of the components assigned to $x$ is at least $\ell$. Then $x$ together with the components assigned to it form a set of size at least $\ell+1$ and have to contain a vertex from the solution. Since we obtain such a connected set for each $x \in X$, the solution has to contain at least $|X|$ vertices from $X \cup Y$. Again we remark that this definition of a reducible pair is local to this section, and not the one we actually end up using. 

To answer (b) we first try to use the $q$-Expansion Lemma (see~\cite{pc_book}), a tool that has found many uses in kernelization. Roughly speaking the Expansion Lemma says the following: if $q \geq 1$ is an integer and $H$ is a bipartite graph with bipartition $(A, B)$ and $B$ is at least $q$ times larger than $A$, then one can find a subset $X$ of $A$ and a subset $Y$ of $B$ such that $N(Y) \subseteq X$, and an assignment of each vertex $y \in Y$ to a neighbor $x$ of $y$, such that every vertex $x$ in $X$ has at least $q$ vertices in $Y$ assigned to it. 

Suppose now that the graph does have an $\ell$-\textsf{COC} solution $S$ of size at most $k$, and that $V(G) \setminus S$ is sufficiently large compared to $S$. The idea is to apply the Expansion Lemma to the bipartite graph $H$, where the $A$ side of the bipartition is $S$ and the $B$ side has one vertex for each connected component of $G - S$. We put an edge in $H$ between a vertex $v$ in $S$ and a vertex corresponding to a component $C$ of $G - S$ if there is an edge between $v$ and $C$ in $G$. If $G - S$ has at least $|S| \cdot \ell$ connected components, we can apply the $\ell$-Expansion Lemma on $H$, and obtain a set $X \subseteq S$, and a collection ${\cal Y}$ of connected components of $G - X$ satisfying the following properties. Every component $C \in {\cal Y}$ satisfies $N(C) \subseteq X$ and $|C| \leq \ell$. Furthermore, there exists an assignment of each connected component $C$ to a vertex $x \in N(C)$, such that every $x \in X$ has at least $\ell$ components assigned to it. Since $x$ has at least $\ell$ components assigned to it, the total size of the components assigned to $x$ is at least $\ell$. But then, $X$ and $Y = \bigcup_{C \in {\cal Y}} C$ form a reducible pair, giving an answer to question (b). Indeed, this argument can be applied whenever the number of components of $G - S$ is at least $\ell \cdot |S|$. Since each component of $G - S$ has size at most $\ell$, this means that the argument can be applied whenever $|V(G) \setminus S| \geq \ell^2 \cdot |S| \geq \ell^2k$. 

Clearly this argument fails to yield a kernel of size $2\ell k$, because it is only applicable when $|V(G)| = \Omega(\ell^2k)$.
At this point we note that the argument above is extremely wasteful in one particular spot: we used the {\em number} of components assigned to $x$ to lower bound the {\em total size} of the components assigned to $x$. To avoid being wasteful, we prove a new variant of the Expansion Lemma, where the vertices on the $B$ side of the bipartite graph $H$ have non-negative integer weights. This new Weighted Expansion lemma states that if $q, W \geq 1$ are integers, $H$ is a bipartite graph with bipartition $(A, B)$, every vertex in $B$ has a non-negative integer weight which is at most $W$, and the total weight of $B$ is at least $(q+W-1) \cdot |A|$, then one can find a subset $X$ of $A$ and a subset $Y$ of $B$ such that $N(Y) \subseteq X$, and an assignment of each vertex $y \in Y$ to a neighbor $x$ of $y$, such that for every vertex in $X$, the total weight of the vertices assigned to it is at least $q$. The proof of the Weighted Expansion Lemma is based on a combination of the usual, unweighted Expansion Lemma with a variant of an argument by Bez\'{a}kov\'{a} and Dani~\cite{BezakovaD05} to round the linear program for Max-min Allocation of goods to customers.

Having the Weighted Expansion Lemma at hand we can now repeat the argument above for proving the existence of a reducible pair, but this time, when we build $H$, we can give the vertex corresponding to a component $C$ of $G - S$ weight $|C|$, and apply the Weighted Expansion Lemma with $q = \ell$ and $W = \ell$. Going through the argument again, it is easy to verify that this time the existence of a reducible pair is guaranteed whenever $|V(G) \setminus S| \leq (2\ell - 1)k$, that is when $|V(G)| \geq 2\ell k$.

We are now left with question (c) - the issue of how to {\em find} a reducible pair in polynomial time. Indeed, the proof of existence crucially relies on the knowledge of an (optimal) solution $S$. To find a reducible pair we use the linear programming relaxation of the $\ell$-\textsf{COC} problem. We prove that an optimal solution to the LP-relaxation has to highlight every reducible pair $(X,Y)$, essentially by always setting all the variables corresponding to $X$ to $1$ and the variables corresponding to $Y$ to $0$. For {\sc Vertex Cover} (i.e  $1$-\textsf{COC}), the classic Nemhauser Trotter Theorem~\cite{NemhauserT74} implies that we may simply include all the vertices whose LP variable is set to $1$ into the solution $S$. For $\ell$-\textsf{COC} with $\ell \geq 2$ we are unable to prove the corresponding statement. We are however, able to prove that if a reducible pair $(X,Y)$ exists, then $X$ (essentially) has to be assigned $1$ and $Y$ (essentially) has to be assigned $0$. We then give a polynomial time algorithm that extracts $X$ and $Y$ from the vertices assigned $1$ and $0$ respectively by the optimal linear programming solution. Together, the arguments (b) and (c) yield the kernel with $2\ell k$ vertices.
We remark that to the best of our knowledge, after the kernel for Vertex Cover~\cite{ChenKJ01} our kernel is the first example of a kernelization algorithm based on linear programming relaxations.

\medskip
\noindent
{\bf Overview of the paper.} In Section~\ref{sec:prelim} we recall basic definitions and set up notations. The kernel for $\ell$-\textsf{COC} is proved in Sections~\ref{sec:maxmin},~\ref{sec:expansion} and~\ref{sec:kernel}. In Section~\ref{sec:maxmin} we prove the necessary adjustment of the results on Max-Min allocation of Bez\'{a}kov\'{a} and Dani~\cite{BezakovaD05} that is suitable to our needs. In Section~\ref{sec:expansion} we state and prove our new Weighted Expansion Lemma, and in Section~\ref{sec:kernel} we combine all our results to obtain the kernel. In Section~\ref{sec:oracle} we improve this kernel by providing the separation oracle.
\section{Preliminaries}\label{sec:prelim}

Let $\mathbb{N}$ denote the set of positive integers $\{0,1,2,\dots\}$. For any non-zero $t\in \mathbb{N}$, $[t]:=\{1,2,\dots,t\}$. We denote a constant function $f:X\to \mathbb{N}$ such that for all $x\in X,  f(x)=c$, by $f=c$. For any function $f:X\to \mathbb{N}$ and a constant $c\in \mathbb{N}$, we define the function $f+c:X\to \mathbb{N}$ such that for all $x\in X, (f+c)(x)=f(x)+c$. We use the same symbol $f$ to denote the restriction of $f$ over a subset of it's domain, $X$.
For a set $\{v\}$ containing a single element, we simply write $v$. A vertex $u\in V(G)$ is said to be incident on an edge $e\in E(G)$ if $u$ is one of the endpoints of $e$. A pair of edges $e,e'\in E(G)$ are said to be adjacent if there is a vertex $u\in V(G)$ such that $u$ is incident on both $e$ and $e'$. For any vertex $u\in V(G)$, by $N(u)$ we denote the set of neighbors of $u$ i.e. $N(u):=\{v\in V(G)\mid uv\in E(G)\}$.
For any subgraph $X\subseteq G$, by $N(X)$ we denote the set of neighbors of vertices in $X$ outside $X$, i.e. $N(X):=(\bigcup_{u\in X}N(u))\setminus X$. A pair of vertices $u,v\in V(G)$ are called \emph{twins} if $N(u)=N(v)$. An induced subgraph on $X\subseteq V(G)$ is denoted by $G[X]$.

A \emph{path} $P$ is a graph, denoted by a sequence of vertices $v_1v_2\dots v_t$ such that for any $i,j\in [t], v_iv_j\in E(P)$ if and only if $|i-j|=1$. A \emph{cycle} $C$ is a graph, denoted either by a sequence of vertices $v_1v_2\dots v_t$ or by a sequence of edges $e_1e_2\dots e_t$, such that for any $i,j\in [t]$ $u_iu_j\in E(C)$ if and only if $|i-j|=1\mod ~t$ or in terms of edges, for any $i,j\in [t]$, $e_i$ is adjacent to $e_j$ if and only if $|i-j|=1\mod t$. The \emph{length} of a path(cycle) is the number of edges in the path(cycle). A \emph{triangle} is a cycle of length $3$. In $G$, for any pair of vertices $u,v \in V(G)$ \textsf{dist}$(u,v)$ represents the length of a shortest path between $u$ and $v$. A \emph{tree} is a connected graph that does not contain any cycle. A \emph{rooted tree} $T$ is a tree with a special vertex $r$ called the root of $T$. With respect to $r$, for any edge $uv\in E(T)$ we say that $v$ is a child of $u$ (equivalently $u$ is parent of $v$) if \textsf{dist}$(u,r)<$\textsf{dist}$(v,r)$. A \emph{forest} is a collection of trees. A \emph{rooted forest} is a collection of rooted trees. 
A \emph{clique} is a graph that contains an edge between every pair of vertices. A vertex cover of a graph is a set of vertices whose removal makes the graph edgeless.

\smallskip
\noindent
{\bf Fixed Parameter Tractability.} A {\em parameterized problem} $\Pi$ is a subset of $\Sigma^* \times \mathbb{N}$. A parameterized problem $\Pi$ is said to be \emph{fixed parameter tractable}(\textsc{FPT}) if there exists an algorithm that takes as input an instance $(I, k)$ and decides whether $(I, k) \in \Pi$ in time $f(k)\cdot n^c$, where $n$ is the length of the string $I$, $f(k)$ is a computable function depending only on $k$ and $c$ is a constant independent of $n$ and $k$. 

A \emph{kernel} for a parameterized problem $\Pi$ is an algorithm that given an instance $(T,k)$ runs in time polynomial in $|T|$, and outputs an instance $(T',k')$ such that $|T'|,k' \leq g(k)$ for a computable function $g$ and $(T,k) \in \Pi$ if and only if $(T',k') \in \Pi$. For a comprehensive introduction to \textsc{FPT} algorithms and kernels, we refer to the book by Cygan et al.~\cite{pc_book}.

A \emph{data reduction rule}, or simply, reduction rule, for a parameterized problem $Q$ is a function $\phi:\Sigma^*\times\mathbb{N}\to \Sigma^*\times\mathbb{N}$ that maps an instance $(I,k)$ of $Q$ to an equivalent instance $(I',k')$ of $Q$ such that $\phi$ is computable in time polynomial in $|I|$ and $k$. We say that two instances of $Q$ are \emph{equivalent} if $(I,k)\in Q$ if and only if $(I',k')\in Q$; this property of the reduction rule $\phi$, that it translates an instance to an equivalent one, is referred as the \emph{safeness} of the reduction rule.
\section{Max-min Allocation}\label{sec:maxmin}

We will now view a bipartite graph $G:=((A,B),E)$ as a relationship between ``customers'' represented by the vertices in $A$ and ``items'' represented by the vertices in $B$. If the graph is supplied with two functions $w_a : A \to \mathbb{N}$ and $w_b : B \to \mathbb{N}$, we treat these functions as a ``demand function'' and a ``capacity'' function, respectively. That is, we consider each item $v \in B$ to have value $w_b(v)$, and every customer $u \in A$ wants to be assigned items worth at least $w_a(u)$. An edge between $u \in A$ and $v \in B$ means that the item $v$ can be given to $u$.

A weight function $f : E(G) \to \mathbb{N}$ describes an assignment of items to customers, provided that the items can be ``divided'' into pieces and the pieces can be distributed to different customers. However this ``division'' should not create more value than the original value of the items. Formally we say that the weight function {\em satisfies} the capacity constraint $w_b(v)$ of $v \in B$ if $\sum_{uv \in E(G)} f(uv) \leq w_b(v)$. The weight function satisfies the capacity constraints if it satisfies the capacity constraints of all items $v \in B$. 

For each item $u \in A$, we say that $f$ {\em allocates} $\sum_{uv \in E(G)} f(uv)$ value to $u$. The weight function $f$ {\em satisfies} the demand $w_a(u)$ of $u \in A$ if it allocates at least $w_a(u)$ value to $u$,
and $f$ satisfies the demand constraints if it does so for all $u \in A$. In other words, the weight function satisfies the demands if every customer gets items worth at least her demand. The weight function $f$ {\em over-satisfies} a demand constraint $w_a(u)$ of $u$ if it allocates strictly more than $w_a(u)$  to $u$. 

We will also be concerned with the case where items are indivisible. In particular we say that a weight function $f : E(G) \to \mathbb{N}$ is {\em unsplitting} if for every $v \in B$ there is at most one edge $uv \in E(G)$ such that $f(uv) > 0$. The essence of the next few lemmas is that if we have a (splitting) weight function $f$ of items whose value is at most $W$, and $f$ satisfies the capacity and demand constraints, then we can obtain in polynomial-time an unsplitting weight function $f'$ that satisfies the capacity constraints and violates the demand constraints by at most $(W-1)$. In other words we can make a splitting distribution of items unsplitting at the cost of making each customer lose approximately the value of the most expensive item. 

Allocating items to customers in such a way as to maximize satisfaction is well studied in the literature. The lemmata~\ref{forrest} and~\ref{stars} are very similar, both in statement and proof, to the work of Bez\'{a}kov\'{a} and Dani~\cite{BezakovaD05}[Theorem 3.2], who themselves are inspired by Lenstra et al.~\cite{LenstraST90}. However we do not see a way to directly use the results of Bez\'{a}kov\'{a} and Dani~\cite{BezakovaD05}, because we need a slight strengthening of (a special case of) their statement.

\begin{lemma}\label{forrest}
There exists a polynomial-time algorithm that given a bipartite graph $G$, a capacity function $w_b : B \to \mathbb{N}$, a demand function $w_a : A \to \mathbb{N}$ and a weight function $f: E(G)\to \mathbb{N}$ that satisfies the capacity and demand constraints, outputs a function $f': E(G)\to \mathbb{N}$ such that $f'$ satisfies the capacity and demand constraints and the graph $G_{f'} = (V(G), \{uv \in E(G) \mid f'(uv) > 0\}) $ 
induced on the non-zero weight edges of $G$ 
is a forest.
\end{lemma}

\begin{proof}
We start with $f$ and in polynomially many steps, change $f$ into the required function $f'$. If $G_{f} = (V(G), \{uv \in E(G) \mid f(uv) > 0\})$ is a forest, then we return $f'=f$. Otherwise, suppose that $G_{f}$ contains a cycle $C:=e_1e_2e_3\dots e_{2s}$. Proceed as follows. Without loss of generality, suppose $c = f(e_1) = min\{f(e) \mid e\in C\}$, and note that $c > 0$. Compute the edge weight function $f^\star:E\to \mathbb{R}$ defined as follows. For $e_i \in C$, we define $f^\star(e_i) = f(e)-c$ if $i$ is odd, and define $f^\star(e_i) = f(e)+c$ if $i$ is even. For $e \notin C$ we define $f^\star(e_i) = f(e)$.

Every vertex of $G$ is incident to either $0$ or exactly $2$ edges of $C$. If the vertex $v$ is incident to two edges of $C$ then one of these edges, say $e_{2i}$, has even index in $C$, and the other, $e_{2i+1}$ has odd. For the edge $e_{2i}$ we have $f^\star(e_{2i}) = f(e_{2i})+c$ and for $e_{2i + 1}$ we have $f^\star(e_{2i + 1}) = f(e_{2i + 1}) - c$.  Thus we conclude that for all $v\in V(G)$, $\sum_{u\in N(v)}f^\star(uv)=\sum_{u\in N(v)}f(uv)$, and that therefore $f^\star$ satisfies the capacity and demand constraints. Furthermore at least one edge that is assigned non-zero weight by $f$ is assigned $0$ by $f^\star$ and $G_{f^\star} = (V(G), \{uv \in E(G) \mid f^\star(uv) > 0\})$ has one less cycle than $G_f$. For a polynomial-time algorithm, repeatedly apply the process described above to reduce the number of edges with non-zero weight, as long as $G_{f^\star}$ contains a cycle.
\end{proof}


\begin{lemma}\label{stars}
There exists a polynomial-time algorithm with the following specifications. It takes as input a bipartite graph $G:=((A,B),E)$, a demand function $w_a:A\to \mathbb{N}$, a capacity function $w_b:B\to \mathbb{N}$, an edge weight function $f:E(G)\to \mathbb{N}$ that satisfies both the capacity and demand constraints, and a vertex $r \in A$. The algorithm outputs an unsplitting edge weight function $h: E(G)\to \mathbb{N}$ that satisfies the capacity constraints, satisfies the demands $w_a' = w_a - (W-1)$ where $W=\max_{v\in B}w_b(v)$, and additionally satisfies the demand $w_a(r)$ of $r$.
\end{lemma}
\begin{proof}
Without loss of generality the graph $G_{f} := (V(G), \{uv \in E(G) \mid f(uv) > 0\})$ is a forest. If it is not, we may apply Lemma~\ref{forrest} to $f$, and obtain a function $f'$ that satisfies the capacity and demand constraints, and such that $G_{f'} = (V(G), \{uv \in E(G) \mid f'(uv) > 0\})$ is a forest. We then rename $f'$ to $f$. By picking a root in each connected component of $G_f$ we may consider $G_f$ as a rooted forest. We pick the roots as follows, if the component contains the special vertex $r$, we pick $r$ as root. If the component does not contain $r$, but contains at least one vertex $u \in A$, we pick that vertex as the root. If the component does not contain any vertices of $A$ then it does not contain any edges and is therefore a single vertex in $B$, we pick that vertex as root. Thus, every item $v \in B$ that is incident to at least one edge in $G_f$ has a unique {\em parent} $u \in A$ in the forest $G_f$. We define the new weight function $h$. For every edge $uv \in E(G)$ with $u \in 
A$ and $v \in B$ we define $h(uv)$ as follows.
	\begin{figure}[!h]
	\includegraphics[scale=0.6]{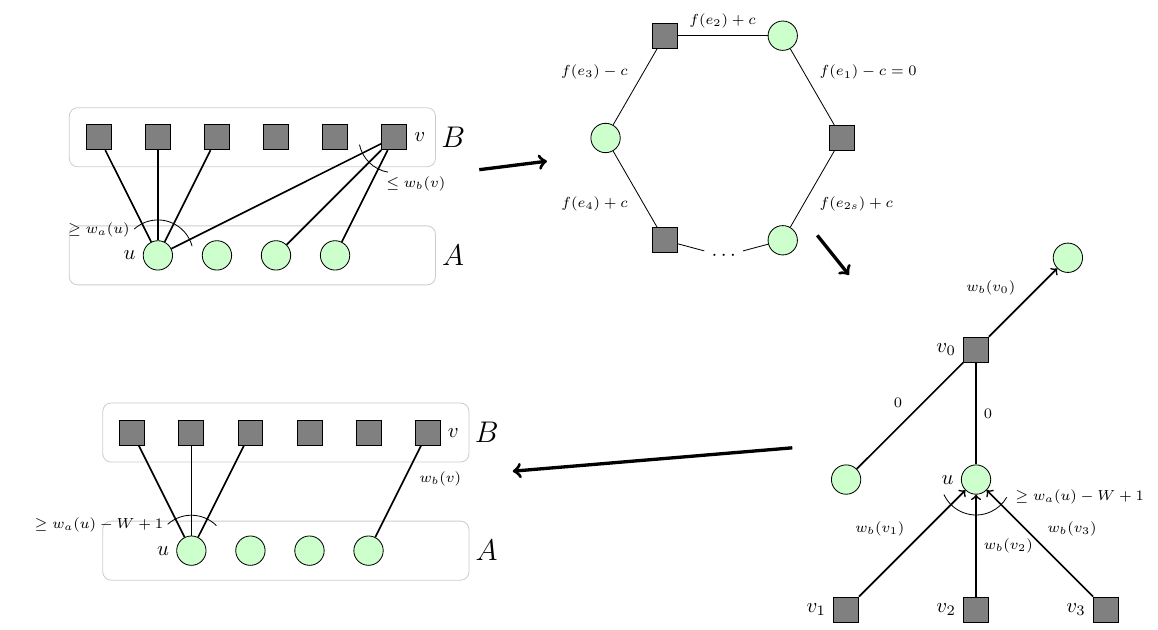}	
	\caption{Proof of Lemma \ref{forrest} and \ref{stars}. Cyclically shift smallest weight in a non-zero weight cycle to obtain a forest. Root each tree in the forest at a vertex in $A$ such that each vertex in $B$ has a parent in $A$. Assign the value of $v\in B$ to its parent $u\in A$. In this new assignment, a non-root vertex $u\in A$ \emph{loses} its parent $v_0\in B$ and $f(v_0u)\leq W-1$ which explains the cost of making a splitting assignment into an unsplitting assignment.}
	\end{figure}
%
%
%
$h(uv) = w_b(v)$ if $u$ is the parent of $v$ in $G_f$, and $h(uv) = 0$ otherwise.

Clearly $h$ is unsplitting and satisfies the capacity constraints. We now prove that $h$ also satisfies the demand constraints $w_a'$ and satisfies the demand constraint $w_a(r)$ of $r$. Consider the demand constraint $w_a'(u)$ for an arbitrary customer $u \in A$. There are two cases, either $u$ is the root of the component of $G_f$ or it is not. If $u$ is the root, then for every edge $uv \in E(G)$ such that $f(uv) > 0$ we have that $uv \in E(G_f)$ and consequently that $u$ is the parent of $v$. Hence $h(uv) = w_b(v) \geq f(uv)$, and therefore $h$ satisfies the demand $w_a(u)$ of $u$. Since $w_a(u) \geq w_a'(u)$, we have that $h$ satisfies the demand $w_a'(u)$. Furthermore, since $r$ is the root of its component this also proves that $h$ satisfies the demand $w_a(r)$.

Consider now the case that $u$ is not the root of its component in $G_f$. Then $u$ has a unique parent in $G_f$, call this vertex $v^\star \in B$. We first prove that $f(uv^\star) \leq w_b(v^\star) - 1$. Indeed, since $v^\star$ is incident to the edge $uv^\star$ we have that $v^\star$ has a parent $u^\star$ in $G_f$, and that $u^\star \neq u$ because $v^\star$ is the parent of $u$. We have that $f(u^\star v^\star) + f(uv^\star) \leq w_b(v^\star)$ and that $f(u^\star v^\star) \geq 1$, because $u^\star v^\star$ is an edge in $G_f$. It follows that $f(uv^\star) \leq w_b(v^\star) - 1$. We now proceed to proving that $h$ satisfies the demand $w_a'(u)$.

For every edge $uv \in E(G) \setminus \{uv^\star\}$ such that $uv \in E(G)$ such that $f(uv) > 0$ we have that $uv \in E(G_f)$ and consequently that $u$ is the parent of $v$. Hence we have that $h(uv) = w_b(v) \geq f(uv)$. Furthermore $h(uv^\star) = 0$ while $f(uv^\star) \leq w_b(v^\star) - 1 \leq W - 1$. Therefore $h$ satisfies the demand $w_a'(u)$.
\end{proof}
\section{The Weighted Expansion Lemma}\label{sec:expansion}
Our kernelization algorithm will use ``$q$-expansions'' in bipartite graphs, a well known tool in kernelization~\cite{pc_book}. We begin by stating the definition of a $q$-expansion and review the facts about them that we will use.
\begin{definition}[\textbf{$q$-expansion}]
	Let $G:=((A,B),E)$ be a bipartite graph. We say that $A$ has $q$-expansion into $B$ if there is a family of sets $\{V_a\mid V_a\subseteq N(a),|V_a|\geq q,a\in A\}$ such that for any pair of vertices $a_i,a_j\in A$,$i\neq j$, $V_{a_i}\cap V_{a_j}=\emptyset$.
\end{definition}

\begin{definition}[\textbf{Twin graph}]\label{def:twing}
For a bipartite graph $G:=((A,B),E)$ with a weight function $w_b:B\to \mathbb{N}$, the twin graph $T_{AB}:=(A,B')$ of $G$ is obtained as follows: $B'$ contains $|w_b(v)|$ twins of every vertex $v\in B$ i.e. $B':=\{v_1,v_2,\dots v_{w_b(v)}\mid v\in B\}$ and edges in $T_{AB}$ such that for all $v\in B$ and $i\in[w_b(v)], N(v_i)=N(v)$ i.e. $E(T_{AB}):=\{av_i|a\in A,v_i\in B',v\in B,av\in E(G)\}$. 
\end{definition}

\begin{lemma}\label{q}\cite{pc_book}
Let $G$ be a bipartite graph with bipartition $(A,B)$. Then there is a $q$-expansion from $A$ into $B$ if and only if $|N(X)|\geq q|X|$ for every $X\subseteq A$. Furthermore, if there is no $q$-expansion from $A$ into $B$, then a set $X\subseteq A$ with $|N(X)|<q|X|$ can be found in polynomial-time.
\end{lemma}
\begin{lemma}[Expansion Lemma~\cite{pc_book}]\label{expansion}
Let $q\geq 1$ be a positive integer and $G$ be a bipartite graph with vertex bipartition $(A,B)$ such that $|B|\geq q|A|$, and
there are no isolated vertices in $B$. Then there exist nonempty vertex sets $X\subseteq A$ and $Y\subseteq B$ such that there is a $q$-expansion of $X$ into $Y$, and no vertex in $Y$ has a neighbor outside $X$, i.e. $N(Y)\subseteq X$. Furthermore, the sets $X$ and $Y$ can be found in time polynomial in the size of $G$.
\end{lemma}
\begin{lemma}[folklore]\label{qset}
	There exists a polynomial-time algorithm that given a bipartite graph $G:=((A,B),E)$ and an integer $q$ decides (and outputs in case yes) if there exist sets $X\subseteq A, Y\subseteq B$ such that there is a $q$-expansion of $X$ into $Y$. 
\end{lemma}
\begin{proof}
We describe a recursive algorithm.  
If $A=\emptyset$ or $B=\emptyset$, then output \textsc{no} and terminate. Otherwise, construct the twin graph $T_{BA}$ with weight function $w:A\to \mathbb{N}$ where for all $u\in A,w(u)=q$ and let $M$ be a maximum matching in $T_{BA}$. Consider the graph $G':=(A,B)$ with edge set $E(G'):=\{uv,u\in A,v\in B\mid u_iv\in M \}$. Let $A'\subseteq A$ such that for all $u\in A', d_{G'}(u)\geq q$ and let $B'\subseteq B$ such that $B':= \bigcup_{u\in A'}N_{G'}(u)$. If $N(B')\subseteq A'$, then return $(A',B')$ and terminate. Otherwise, recurse on $G[A'\cup (B\setminus N_G(A\setminus A'))]$.

 If there are no sets $X,Y$ such that there is a $q$-expansion of $X$ into $Y$, then for any pair of sets $A'\subseteq A,B'\subseteq B$ either $N(B')\setminus A'\neq \emptyset$ or $|B'|<q|A'|$. Since at each recursive step, the size of the graph with which the algorithm calls itself decreases, eventually either $A'$ becomes empty or $B\setminus N_G(A\setminus A')$ becomes empty. Hence, the algorithm outputs no. Now we need to show that if there exist sets $(A^*,B^*)$ such that there is a $q$-expansion of $A^*$ into $B^*$, then at each recursive call, we have that $A^*\subseteq A$ and $B^*\subseteq B$. At the start of the algorithm, $A^*\subseteq A$ and $B^*\subseteq B$. Since $N(B^*)\subseteq A^*$ and for all $u\in A^*$ $d_G(u)\geq q$, we have that $A^*\cup B^*\subseteq V(G')$. If $N(B')\subseteq A'$, then the algorithm of Lemma \ref{expansion} when run on $G',q$ will output $(A^*,B^*)$. 
 Note that $B^*\subseteq B'$. At the recursive step, $A^*\subseteq A'$ and since $B^*\cap N_G(A\setminus A')=\emptyset$, we have that $B^*\subseteq B'\setminus N_G(A\setminus A')$. Hence, $G[A^*\cup B^*]$ is a subgraph of $G[A'\cup (B\setminus N_G(A\setminus A'))]$ which concludes the correctness of the algorithm. Since at each recursive call the size of the graph decreases by at least 1, the total time taken by the above algorithm is polynomial in $n$.   
\end{proof}

One may think of a $q$-expansion in a bipartite graph with bipartition $(A,B)$ as an allocation of the items in $B$ to each customer in $A$ such that every customer gets at least $q$ items. For our kernel we will need a generalization of $q$-expansions to the setting where the items in $B$ have different values, and every customer gets items of total value at least $q$.

\begin{definition}[\textbf{Weighted $q$-expansion}] Let $G:=((A,B),E)$ be a bipartite graph with capacity function $w_b : B \to \mathbb{N}$. Then, a weighted $q$-expansion in $G$ is an edge weight function $f : E(G)\to \mathbb{N}$ that satisfies the capacity constraints $w_b$ and also satisfies the demand constraints $w_a = q$. 
For an integer $W \in \mathbb{N}$, the $q$-expansion $f$ is called a $W$-{\em strict} $q$-expansion if $f$ allocates at least $q+W-1$ value to at least one vertex $r$ in $A$, and in this case we say that $f$ is $W$-strict at $r$. Further, a $q$-expansion $f$ is {\em strict} (at $r$) if it is $1$-strict (at $r$). 
If $f$ is unsplitting we call $f$ an {\em unsplitting} $q$-expansion.
\end{definition}

\begin{lemma}\label{qexp}
There exists a polynomial-time algorithm that given a bipartite graph $G:=((A,B),E)$, an integer $q$ and a capacity function $w_b : B\to \mathbb{N}$ outputs (if it exist) two sets $X\subseteq A$ and $Y\subseteq B$ along with a weighted $q$-expansion in $G[X\cup Y]$ such that $N(Y)\subseteq X$.
\end{lemma}
\begin{proof}
Construct the twin graph $T_{AB}:=(A,B')$ of $G$. Run the algorithm of Lemma \ref{qset} with input $T_{AB},q$ that outputs sets $X\subseteq A$ and $Y'\subseteq B'$ such that $X$ has $q$-expansion into $Y'$ and $N(Y')\subseteq X$. Consider the set $Y:=\{v\in B\mid v_i\in Y'\}$. Define a weight function $f:E(G[X\cup Y])\to \mathbb{N}$ as follows: for all $uv\in E(G[X\cup Y])$ $f(uv)=|\{v_i\in Y'|v_i \text{ matched to }u\}|$. 

Clearly, $N(Y)\subseteq X$. Now we claim that $f$ is a weighted $q$-expansion in $G[X\cup Y]$ with capacity function $w_b$ and demand function $w_a=q$. For any vertex $u\in A$, there are at least $q$ vertices in $Y'$ are matched to $u$. Hence for all $u\in A$, we have that $\sum_{v\in N(u)}f(uv)\geq q=w_a$. At the same time, for any vertex $v\in B$, there are at most $w_b(v)$ copies of $v$ in $Y'$. Therefore, for all $v\in Y$ we have $\sum_{u\in N(v)}f(uv)\leq w_b(v)$. 
\end{proof}
\begin{lemma}\label{qstrict}
	There exists a polynomial-time algorithm that given a weighted $q$-expansion $f:E(G)\to \mathbb{N}$ in $G:=((A,B),E)$, a capacity function $w_b:B\to \mathbb{N}$ and an integer $W$ such that $W=\max_{e\in E(G)}f(e)$ outputs an unsplitting $W$-strict weighted $(q-W+1)$-expansion in $G$. 
\end{lemma}
\begin{proof}
Run the algorithm of Lemma \ref{stars} with inputs $G,f,w_a=q,w_b,W$ and a vertex $u\in A$. In case $f$ is strict, $u$ is the vertex $r$ that makes $f$ strict. Let the function $h:E(G)\to \mathbb{N}$ be the output of Lemma \ref{stars}. Now $h$ is an unsplitting edge weight function that satisfies the capacity constraints, satisfies the demands $q - W + 1$, and additionally satisfies the demand $q$ of $u$. Hence, $h$ is the required unsplitting weighted $W$-strict $(q-W+1)$-expansion in $G$.
\end{proof}
  
\begin{lemma}[Weighted Expansion Lemma]\label{wexpansion}
Let $q, W \geq 1$ be positive integers and $G$ be a bipartite graph with vertex bipartition $(A,B)$ and $w_b : B \to \{1,\ldots, W\}$ be a capacity function such that $\sum_{v \in B} w_b(v) \geq (q+W-1)\cdot|A|$, and there are no isolated vertices in $B$. Then there exist nonempty vertex sets $X\subseteq A$ and $Y\subseteq B$ such that $N(Y)\subseteq X$ and there is an unsplitting weighted $W$-strict $q$-expansion of $X$ into $Y$. Furthermore, the sets $X$ and $Y$ can be found in time polynomial in the size of $G$.
\end{lemma}  
  
\begin{proof}
Construct the twin graph $T_{AB}$ from $G$ and $w_b$, the bipartition of $T_{AB}$ is $(A,B')$. Now, obtain using the Expansion Lemma~\ref{expansion} with $q' = q + W - 1$ on $T_{AB}$ sets $X \subseteq A$ and $Y' \subseteq B'$, such that $N(Y') \subseteq X$ and there is a $(q+W-1)$-expansion from $X$ to $Y'$ in $T_{AB}$.

Let $Y:=\{v\in B\mid v_i\in Y'\}$ (here the $v_i \in Y'$ are as in Definition~\ref{def:twing}). Then $N(Y) \subseteq X$ and the  $(q+W-1)$-expansion from $X$ to $Y'$ in $T_{AB}$ immediately yields a weighted $(q+W-1)$-expansion $f$ from $X$ to $Y$ in $G$. Applying Lemma~\ref{qstrict} on $G[X \cup Y]$ using the weighted $(q+W-1)$-expansion $f$ proves the statement of the lemma.
\end{proof}

\section{Obtaining the Linear Kernel}\label{sec:kernel}

\begin{definition}
For a graph $G$ and a pair of vertex-disjoint sets $X,Y\subseteq V(G)$, we define the weighted graph $\tilde{G}_{XY}$ as follows:
$V(\tilde{G}_{XY}) := X\cup \tilde{Y}$ such that there is a bijection $h:cc(G[Y])\to \tilde{Y}$ where $cc(G[Y])$ is the set of connected components of $G[Y]$. $E(\tilde{G}_{XY}) :=\{xc\mid x\in X,c\in \tilde{Y}, c=h(C) \text{ and } x\in N_G(C)\}$. We also define a weight function $w:\tilde{Y}\to \mathbb{N}$ such that for all $c\in \tilde{Y}, w(c)=|h^{-1}(c)|$.
\end{definition}
\begin{definition}[\textbf{Reducible Pair}]
For a graph $G$, a pair of vertex-disjoint sets $(X,Y)$ where $X,Y\subseteq V(G)$ is called a  (strict) reducible pair if $N(Y)\subseteq X$, the size of every component in $G[Y]$ is at most $\ell$, and there exists a (strict) weighted $(2\ell-1)$-expansion in $\tilde{G}_{XY}$.
\end{definition}
\begin{definition}
	A reducible pair $(X,Y)$ is called minimal if there is no reducible pair $(X',Y')$ such that $X'\subset X$ and $Y'\subseteq Y$.
\end{definition}
\begin{lemma}\label{biker}
	There exists a polynomial-time algorithm that given an $\ell$-COC instance $(G,k)$ together with a vertex-disjoint set pair $A,B\subseteq V(G)$ outputs (if it exists) a reducible pair $(X,Y)$ where $X\subseteq A$ and $Y\subseteq B$.
\end{lemma}
\begin{proof}
Construct $\tilde{G}_{AB}:=(A,\tilde{B})$ and run the algorithm of Lemma \ref{qexp} with input $\tilde{G}_{AB},w,q=2\ell-1$ which outputs sets $X\subseteq A$ and $Y'\subseteq \tilde{B}$ (if it exists) along with a weighted $(2\ell-1)$-expansion of $X$ into $Y'$ such that $N(Y')\subseteq X$. Now from $Y'$ we obtain the set $Y:=\bigcup_{y\in Y'}h^{-1}(y)$. Clearly, $N(Y)\subseteq X$ and hence, $(X,Y)$ is the desired reducible pair.
\end{proof}
\begin{lemma}\label{redexist}
	Given an $\ell$-COC instance $(G,k)$, if $|V(G)|\geq 2\ell k$ and $(G,k)$ is a \textsc{yes}-instance, then there exists a reducible pair $(X,Y)$.
\end{lemma}
\begin{proof}
Without loss of generality, we can assume that $G$ is a connected graph. Let $S$ be an $\ell$-COC solution of size at most $k$. Clearly, $|V\setminus S|\geq (2\ell-1)k$. We define $A:=S$ and $B:=V\setminus S$ and construct $\tilde{G}_{AB}=(A,\tilde{B})$. We have the weight function $w_b:\tilde{B}\to \mathbb{N}$ such that for all $v\in \tilde{B},w_b(v)=|h^{-1}(v)|\leq \ell$, as the size of components in $G[V\setminus S]$ is at most $\ell$. We have that $\sum_{v\in \tilde{B}}w_b(v)\geq (2\ell-1)|A|$ and there are no isolated vertices in $\tilde{B}$. Hence, $(A,B)$ is the desired reducible pair.
\end{proof}
\begin{lemma}\label{pack}
Let $(X,Y)$ be a reducible pair. Then, there exists a partition of $X\cup Y$ into $C_1,...,C_{|X|}$ such that (i) for all $u_i\in X$, we have $u_i\in C_j$ if and only if $i=j$, (ii) for all $i\in [|X|],~|C_i|\geq \ell+1$, (iii) for every component $C$ in $G[Y]$, there exists a unique $C_i$ such that $V(C)\subseteq C_i$ and $u_i\in N(C)$ and (iv) if $(X,Y)$ is a strict reducible pair, then there exists $C_j$ such that $|C_j|\geq 2\ell+1$. 
\end{lemma}
\begin{proof}
Construct $\tilde{G}_{XY}:=(X,\tilde{Y})$. Run the algorithm of Lemma \ref{qstrict} with input $\tilde{G}_{XY},q=2\ell-1$, and $W=\ell$(as the capacity of any vertex in $\tilde{Y}$ is at most $\ell$) which outputs an unsplitting weighted $\ell$-expansion $f'$ in $\tilde{G}_{XY}$. In polynomial time, we modify $f'$ such that if there is a vertex $v\in \tilde{Y}$ such that $\forall u\in N(v), f'(uv)=0$, we choose a vertex $u\in N(v)$ and set $f'(uv)=w_b(v)$. For each $u_i\in X$ define $C_i:=u_i\bigcup_{f'(u_iv)\neq 0}h^{-1}(v)$. Since $f'$ is unsplitting, the collection $C_1,\dots,C_{|X|}$ forms a partition of $X\cup Y$. By the definition of $C_i$, we have that for any $u_i\in X$, $u_i\in C_j$ if and only if $i=j$. For any component $C$ in $G[Y]$, $h(C)$ is matched to a unique vertex $u_i\in X$ by $f'$, we have that $V(C)\subseteq C_i$. As $f'$ is a weighted $\ell$-expansion, $|C_i|=1+\sum_{f'(u_iv)\neq 0}|h^{-1}(v)|=1+\sum_{f'(u_iv)\neq 0}f'(u_iv)\geq 1+\ell$. Let $(X,Y)$ be strict at $u_j\in X$. Then, we can 
use Lemma \ref{qstrict} to obtain the expansion $f'$ such that it is strict at $u_j$. Hence, $|C_j|=1+\sum_{f'(u_jv)\neq 0}|h^{-1}(v)|=1+\sum_{f'(u_jv)\neq 0}f'(u_jv)> 1+\ell+(\ell-1)$ which implies $|C_j|\geq 2\ell +1$. This concludes the proof of the lemma.
\end{proof}
\begin{lemma}\label{X}
	Let $(X,Y)$ be a reducible pair. If $(G,k)$ is a \textsc{yes}-instance for $\ell$-COC, then there exists an $\ell$-COC solution $S$ of size at most $k$ such that $X\subseteq S$ and $S\cap  Y=\emptyset$. 
\end{lemma}
\begin{proof}
By Lemma \ref{pack} we have that there are $C_1,\dots, C_{|X|}\subseteq X\cup Y$ vertex disjoint sets of size at least $\ell+1$ such that for all $i\in [|X|]$, $G[C_i]$ is a connected set. Let $S'$ be an arbitrary solution. Then, $S'$ must contain at least one vertex from each $C_i$. Let $S:=S'\setminus (X\cup Y)\cup X$. We have that $|S|\leq |S'|-|X|+|X|=|S'|$. As any connected set of size $\ell+1$ that contains a vertex in $Y$ also contains a vertex in $X$ and $X\subseteq S$, $S$ is also an $\ell$-COC solution. 
\end{proof}
Now we encode an $\ell$-COC instance $(G,k)$ as an \textsc{Integer Linear Programming} instance. We introduce $n=|V(G)|$ variables, one variable $x_v$ for each vertex $v\in V(G)$. Setting the variable $x_v$ to 1 means that $v$ is in $S$, while setting $x_v=0$ means that $v$ is not in $S$. To ensure that $S$ contains a vertex from every connected set of size $\ell+1$, we can introduce constraints $\sum_{v\in C}x_v\geq 1$ where $C$ is a connected set of size $\ell+1$. The size of $S$ is given by $\sum_{v\in V(G)}x_v$. This gives us the following ILP formulation:\\
$\begin{matrix}
\mbox{minimize} & \sum_{v\in V(G)}x_v,\\
\mbox{subject to} & \sum_{v\in C}x_v\geq 1 & \mbox{for every connected set }C\mbox{ of size }\ell+1\\
& 0\leq x_v\leq 1 &\mbox{ for every }v\in V(G)\\
& x_v\in \mathbb{Z} & \mbox{ for every }v\in V(G).
\end{matrix}$\\
 Note that there are $n^{\mathcal{O}(\ell)}$ connected sets of size at most $\ell$ in a graph on $n$ vertices. Hence, providing an explicit ILP requires $n^{\mathcal{O}(\ell)}$ time which forms the bottleneck for the runtime of the kernelization algorithm that follows. We consider the Linear Programming relaxation of above ILP obtained by dropping the constraint that $x\in \mathbb{Z}$.
By an optimal LP solution $S_L$ with weight $L$ we mean the set of values assigned to each variable,
and optimal value is $L$. 
For a set of vertices $X\in V(G)$, $X=1$ ($X=0$) denotes that every variable corresponding to vertices in $X$ is set to $1$ ($0$).  
\begin{lemma}\label{helper}
	Let $S_L$ be an optimal LP solution for $G$ such that $x_v=1$ for some $v\subseteq V(G)$. Then, $S_L-x_v$ is an optimal LP solution for $G-v$ of value $L-1$.
\end{lemma}
\begin{proof}
	Clearly, $S_L-x_v$ is feasible solution for $G-v$ of value $L-1$. Suppose it is not optimal. Let $S_{L'}$ be an optimal LP solution for $G-v$ such that $L'<L-1$. Then, $S_{L'}\cup x_v$ with $x_v=1$ is an optimal LP solution for $G$ with value $<L-1+1=L$ contradicting that the optimal solution value of LP for $G$ is $L$.
\end{proof}
From now on by running LP after setting $x_v=1$ for some vertex $v$, we mean running the LP algorithm for $G-v$ and including $x_v=1$ in the obtained solution to get a solution for $G$.
\begin{lemma}\label{strictx}
	Let $(X,Y)$ be a strict reducible pair. Then every optimal LP solution sets at least one variable corresponding to a vertex in $X$ to 1.
\end{lemma}
\begin{proof}
By Lemma \ref{X}, we have that every connected set of size $\ell+1$ in $G[X\cup Y]$ contains a vertex in $X$. Hence, from any LP solution $S_L$, a feasible LP solution can be obtained by setting $X=1$ and $Y=0$. Since, we have at least $|X|$ many vertex disjoint LP constraints, for each $v_i\in X$, we have $\sum_{u\in C_i}x_u=1$.
By Lemma \ref{pack}, there is a set $C_j\subseteq X\cup Y$ such that $|C_j|\geq 2\ell+1$. If $x_{v_j}\neq 1$, then there is a vertex $w\in C_j$ such that $x_w>0$. Let $w\in C\subset C_j$ where $G[C]$ is a connected component in $G[Y]$. Since $|C|\leq \ell$, there is a connected set $C'$ of size at least $\ell+1$ in $G[C_j]-C$. But now $\sum_{u\in C'}x_u<1$ contradicting that $S_L$ is feasible.
\end{proof}

\begin{lemma}\label{redpair}
	 Let $(X,Y)$ be a minimal reducible pair. If for any vertex $v\in X$, an optimal LP solution sets $x_v=1$, then it also sets $X=1$ and $Y=0$.
\end{lemma}
\begin{proof}
We prove the lemma by contradiction. Let $X'\subset X$ be the largest subset of $X$ such that $X'=1$.  Consider $\tilde{G}_{XY}$. Let $Y'\subseteq \tilde{Y}$ be the set of vertices such that $N(Y')\subseteq X'$. Let $Z:=\bigcup_{v\in Y'}h^{-1}(v)$. By the minimality of $(X,Y)$, we have that $\sum_{v\in Y'}w(v) < (2\ell-1)|X'|$. Hence, $\sum_{v\in \tilde{Y}\setminus Y'}w(v)>(2\ell-1)|X\setminus X'|$. Clearly, the weighted $(2\ell-1)$-expansion in the reducible pair $(X,Y)$ when restricted to $(X\setminus X',Y\setminus Z)$ provides a weighted $(2\ell-1)$-expansion of $X\setminus X'$ into $Y\setminus Z$. This implies that $(X\setminus X',Y\setminus Z)$ is a strict reducible pair in $G-(X'\cup Z)$. By Lemma \ref{helper}, we have that the LP solution restricted to $G-(X'\cup Z)$ is optimal. Since $(X\setminus X',Y\setminus Z)$ is a strict reducible pair, by Lemma \ref{strictx}, there is a vertex $u\in X\setminus X'$ such that $x_u=1$, but this contradicts the maximality of $X'$. Therefore, if for any vertex $v\in 
X$, an LP solution sets $x_v=1$, then it sets $X=1$ and $Y=0$. 
\end{proof}
\begin{lemma}\label{ker}
	There exists a polynomial time algorithm that given an integer $\ell$ and $\ell$-COC instance $(G,k)$ on at least $2\ell k$ vertices either finds a reducible pair $(X,Y)$ or concludes that $(G,k)$ is a \textsc{no}-instance.
\end{lemma}
\begin{proof} If $(G,k)$ is a \textsc{yes}-instance of $\ell$-COC, then by Lemma \ref{redexist}, there exists a reducible pair $(X,Y)$. We use the following algorithm to find one:
\begin{description}
	\item[Step 1] Run the LP algorithm. Let $A=1$ and $B=0$ in the LP solution.
	\item[Step 2] If both $A$ and $B$ are non-empty, then run the algorithm of Lemma \ref{biker} with input $(G,k),A,B$. If it outputs a reducible pair $(X,Y)$, then return $(X,Y)$ and terminate. Otherwise, go to step 3.	
	\item[Step 3] Now we do a linear search for a vertex in $X$. For each vertex $v\in V(G)$, do the following: in the original LP introduce an additional constraint that sets the value of the variable $x_v$ to $1$ i.e. $x_v=1$ and run the LP algorithm. If the optimal value of the new LP is the same as the optimal value of the original LP, then let $A=1$ and $B=0$ be the sets of variables set to $1$ and $0$ respectively in the optimal solution of the new LP and go to step 2.
	\item[Step 4] Output a trivial \textsc{no}-instance.  	
\end{description}
Step 1 identifies the set of variables set to 1 and 0 by the LP algorithm. By Lemma \ref{redpair}, we have that if there is a minimal reducible pair $(X,Y)$ in $G$, then $X\subseteq A$ and $Y\subseteq B$. So, in Step 2 if the algorithm succeeds in finding one, we return the reducible pair and terminate otherwise we look for a potential vertex in $X$ and set it to 1. If $(X,Y)$ exists, then for at least one vertex, setting $x_v=1$ would set $X=1$ and $Y=0$ (by Lemma~\ref{redpair}) without changing the LP value and we go to Step 2 to find it. If for each choice of $v\in V(G)$, the LP value changes when $x_v$ is set to 1, we can conclude that there is no reducible pair and output a trivial \textsc{no} instance. Since, we need to do this search at most $n$ times and each step takes only polynomial time, the total time taken by the algorithm is polynomial in the input size.
\end{proof}
\begin{theorem}\label{LinearKernel}
	For every constant $\ell\in \mathbb{N}$, $\ell$-\textsc{Component Order Connectivity} admits a kernel with at most $2\ell k$ vertices that takes $n^{\mathcal{O}(\ell)}$ time.
\end{theorem}

\section{Separation oracle for $\ell$-COC}\label{sec:oracle}
For $\ell$-COC, we have an LP with $n^{\ell+1}$ constraints: for every connected set of size $\ell+1$, we need the sum of the variables to be at least $1$. So a separation oracle for this LP should take an input $G$ and variable values $x_v$ on all the vertices and then determine whether there exists a connected set of $\ell+1$ vertices with sum less than $1$. If yes, then that is a violated constraint. Let us first define this problem formally:
%

\smallskip
\noindent
\fbox{\parbox{\textwidth-\fboxsep}{
\textsc{Min $\ell$-Connected-Subgraph ($\ell$-MCS)}\\
\textbf{Input:} A graph $G$ and vertex weights $w:V\to \mathbb{R}$.\\
\textbf{Task:} Find a connected subgraph $C$ on at least $\ell+1$ vertices such that the weight of subgraph $C$ is less than $1$, if one exists where weight of a subgraph is the sum of weights of vertices in it.
}}
\smallskip

Now, we prove that $\ell$-MCS is \NP-complete by reducing from \textsc{Set Cover}.
\begin{theorem}
	\textsc{Min $\ell$-Connected-Subgraph} is \NP-complete.
\end{theorem}
\begin{proof}
	Let $I:=(U,\mathcal{H})$ be an instance of \textsc{Set Cover} such that $|U|=n$ and $|\mathcal{H}|=m$. Construct the set-element incidence (bipartite) graph $G:=(A,B,E)$ for $I$ as follows: $A$ contains a vertex for each set in $\mathcal{H}$ and $B$ contains a vertex for each element in $U$. Two vertices $u\in A,v\in B$ are adjacent if and only if the element corresponding to $v$ belongs to the set corresponding to $u$. In addition, $B$ contains a special vertex $b$ which is adjacent to every vertex in $A$. We define the weight function as follows: $w:V(G)\to \mathbb{R}$ such that $\forall u\in A, w(u)=\frac{1}{k+1}$ and $\forall v\in B, w(v)=0$. Finally, we set $\ell:= n+k$.
	
	We claim that $I$ has a set cover of size $k$ if and only if $G$ has a connected subgraph $C$ of size $\ell+1$ and weight less than $1$. For the forward direction, assume $\mathcal{F}$ is a set cover of size $k$. Let $H\subseteq A$ corresponding to elements in $\mathcal{F}$. Then, $G[H\cup B]$ is a connected set on $k+n+1$ vertices: every vertex in $B-b$ has an edge to at least one vertex in $H$ and $b$ is adjacent to every vertex in $H$.
	
	For the backward direction, let $C$ be a connected graph on $n+k+1$ vertices and weight less than $1$. Clearly, $|A\cap V(C)|< k+1$, otherwise weight of $C$ would be at least $1$. Rest of the vertices must belong to $B$. As $C$ is connected and $B$ is independent, every vertex in $B$ must have an edge in $A\cap V(C)$. Hence, the set corresponding to elements in $A\cap V(C)$ form a set cover of size $k$.
\end{proof}

%
%
%
%

To find a connected subgraph of minimum weight, we use color coding. The technique of color coding was introduced by Alon, Yuster and Zwick~\cite{alon1995colorcoding}. Suppose the size (number of vertices) in the sought subgraph $H$ be $k$. If we color the vertex set of $G$, an $n$-vertex graph, using $k$ colors where each vertex is assigned one of the $k$ colors uniformly and independently at random, then with probability at least $e^{-k}$, the vertices of $H$ are colored with pairwise distinct colors. To justify this argument it is easy to see that there are $k^n$ possible colorings of $V(G)$ and $k!k^{n-k}$ of these colorings are such that $V(H)$ has pairwise distinct colors which implies the probability of success to be at least $\frac{k^n}{k!k^{n-k}}\geq e^{-k}$. A typical sequence of steps is to first obtain a random coloring of $G$, then find a \emph{colorful} subgraph $H$ in $G$. To improve the success probability, repeat the above steps $O(e^k)$ times.

Now we describe the oracle implementation formally.
Given a graph $G$ and a coloring $c:V(G)\to [k]$, a subgraph $H$ of $G$ is called \emph{colorful} if the vertices in $V(H)$ get pairwise distinct colors under $c$. Note that $c$ need not be a proper coloring of $G$ in which we require endpoints of an edge to have distinct colors. In our case, $H$ is a connected subgraph on $\ell+1$ vertices. As the first step of the algorithm, we color the vertices in $V(G)$ uniformly and independently at random. As argued above, with probability at least $e^{-(\ell+1)}$, a connected set $H$ with minimum total weight gets multicolored. Now our task is to find a minimum weight multicolored connected set for which we use dynamic programming.  
 
\begin{theorem}\label{colorfulCOC}
	Let $G$ be an undirected, weighted graph with weights $w:V(G)\to \mathbb{R}$ and let $c:V(G)\to [k]$ be a coloring of its vertices with $k$ colors. There exists a deterministic algorithm that checks in time $3^k\pol$ whether $G$ contains a colorful connected subgraph on $k$ vertices and, if this is the case, returns one such subgraph of smallest weight.
\end{theorem}
\begin{proof}
	We assume that $G$ is a connected graph.
	If $G$ is not a connected graph, then the algorithm described below is run on each connected component, and the output of the algorithm is the logical OR of outputs corresponding to each component.

	Define a dynamic programming table $T$ that takes as input a subset of colors $S$ and a vertex $v$ and returns the minimum weight of a multicolored connected subset that uses $v$ and all the colors in $S$ exactly once. If $v$ is not colored with a color in $S$, then set $T[S,v]=\infty$ . Otherwise, 
\begin{equation}\label{eqn}
	T[S,v]= \min_{S'\subseteq S, v'\in N(v)}\{T[S\setminus S',v]+T[S',v']\} 
\end{equation}
In Equation \ref{eqn}, the subset $S'$ is such that the color of $v$ is not in $S'$ and $v'$ is a neighbor of $v$. Finally, the algorithm returns the smallest value in the last row of the table. If this value is $\infty$, then $G$ does not contain any colorful connected set on $k$ vertices. Otherwise, one can recover the multicolored connected set via back-tracking.

To argue about the correctness of the algorithm, it is sufficient to show the correctness of the Equation \ref{eqn}. We need to show two inequalities. First, we show that $T[S,v]\leq \min_{S'\subseteq S, v'\in N(v)}\{T[S\setminus S',v]+T[S',v']\}$. Observe that for any $v'\in N(v)$, the union of connected sets corresponding to $T[S\setminus S',v]$ and $T[S',v']$ is also a connected set and as $S\setminus S'$ and $S'$ are disjoint, the union is colorful as well. Hence, $T[S,v]\leq T[S\setminus S',v]+T[S',v']$. 

For the other inequality, let $H$ be a connected set corresponding to $T[S,v]$. There is a partition $(A,B)$ of $H$ such that $v\in A$, $B$ contains a neighbor of $v$ and both $G[A]$ and $G[B]$ are connected sets: 

Remove $v$ from $H$. If $H\setminus v$ is connected let $A:=v$ and $B:=H\setminus v$. Otherwise, let $A$ be defined as the union of $v$ with all connected components of $H\setminus v$ except one (say, set $C$) and $B$ has the component $C$. 

Let $S'$ be the set of colors of vertices in $B$ and $S\setminus S'$ be the color of vertices in $A$. Note that $v\in A$. Let $v'$ be a neighbor of $v$ in $B$. By definition, $T[S\setminus S',v]\leq w(A)$ and $T[S',v']\leq w(B)$ where $w(A)$ and $w(B)$ are the sum of weights of vertices in $A$ and $B$ respectively. Hence, $T[S,v]=w(A)+w(B)\geq T[S\setminus S',v]+T[S',v']$. This concludes the proof of correctness of the algorithm.

Observe that there are $2^{|S|}\cdot n$ terms in the equation \ref{eqn}. Hence, each of entries $T[S,v]$ can be computed in $2^{|S|}\pol$ time. There are ${k\choose |S|}$ many subsets of colors of size $|S|$. Hence, the running time of the algorithm is bounded by $\sum_{i=0}^{k}{k \choose i} 2^i\pol=3^k\pol$. 
\end{proof}
   
\paragraph{Derandomization.}
Algorithms based on color coding are randomized, but one can often derandomize these algorithms.
The basic idea of derandomization is as follows: instead of picking a random coloring $c:[n]\to [k]$, we deterministically construct a family $\mathcal{F}$ of functions $f:[n]\to [k]$ such that it is guaranteed that one of the functions from $\mathcal{F}$ has the property that we hope to attain by choosing a random coloring $c$.
\begin{definition}
	An $(n,k,\ell)$-splitter $\mathcal{F}$ is a family of functions from $[n]$ to $[\ell]$ such that for every set $S\subseteq [n]$ of size $k$ there exists a function $f\in \mathcal{F}$ that splits $S$ evenly. That is, for every $1\leq j,j'\leq \ell$, $|f^{-1}(j)\cap S|$ and $|f^{-1}(j')\cap S|$ differ by at most $1$.
\end{definition}
\begin{theorem}[\cite{alon1995colorcoding}]
	For any $n,k\geq 1$ one can construct an $(n,k,k^2)$-splitter of size $k^{\Oh(1)}\log n$ in time $k^{\Oh(1)}n\log n$.
\end{theorem}
\begin{definition}
	An $(n,k,k)$-splitter is called an $(n,k)$-perfect hash family.
\end{definition}
\begin{theorem}[\cite{NaorSS95}]\label{perfect}
	For any $n,k\geq 1$ one can construct an $(n,k)$-perfect hash family of size $e^kk^{\Oh(\log k)}\log n$ in time $e^kk^{\Oh(\log k)}n\log n$.
\end{theorem}
Let $(G,k)$ be the input instance for $\ell$-COC, where $n=|V(G)|$. Instead of taking a random coloring $c$ of $V(G)$, we use Theorem \ref{perfect} to construct an $(n,k)$-perfect hash family $\mathcal{F}$. Then, for each $f\in \mathcal{F}$, we invoke the dynamic programming algorithm of Theorem \ref{colorfulCOC} for the coloring $c:=f$. The properties of an $(n,k)$-perfect hash family $\mathcal{F}$ ensure that, if there exists a connected set $H$ on $\ell+1$ vertices in $G$, there there exists $f\in \mathcal{F}$ that is injective on $V(H)$ and, consequently, the algorithm of Theorem \ref{colorfulCOC} finds a colorful subgraph $H$ for the coloring $c:=f$. As a consequence, we have the following improvement in Theorem~\ref{LinearKernel}.
\begin{theorem}
	For every constant $\ell\in \mathbb{N}$, $\ell$-\textsc{Component Order Connectivity} admits a kernel with at most $2\ell k$ vertices that can be computed in $(3e)^{\ell}\cdot\pol$ time.
\end{theorem}


\subparagraph*{Acknowledgements} 
The research leading to these results has received funding from the European Research Council under the European Union's Seventh Framework Programme (FP7/2007-2013) / ERC grant agreement no. 306992 and the Beating Hardness by Pre-processing grant funded by the Bergen Research Foundation.



\bibliography{p-Kumar}{}




\end{document}